\documentclass[
    aps,
    reprint,
    amsmath,
    amssymb,
    superscriptaddress,
    longbibliography
]{revtex4-2}

\usepackage{graphicx}
\usepackage{bm}
\usepackage[colorlinks=true,allcolors=blue]{hyperref}
\usepackage{physics}
\usepackage{enumitem}
\usepackage{amsthm}
\usepackage{bbm}
\usepackage{float}
\makeatletter
\let\newfloat\newfloat@ltx
\makeatother
\usepackage{algorithm}
\usepackage{algpseudocode}

\newtheorem{theorem}{Theorem}
\newtheorem{lemma}{Lemma}

\usepackage[normalem]{ulem}

\begin{document}

\title{Magic-Free Certification of Logical Error Rates in Universal Fault-Tolerant Circuits}

\author{Gebrat Porta-San Miguel}
\email{gebrat.porta.san.miguel@uni-hamburg.de}
\affiliation{Institute for Quantum Physics, University of Hamburg, Luruper Chaussee 149, 22761 Hamburg, Germany}

\author{Shahram Panahiyan}
\affiliation{Institute for Quantum Physics, University of Hamburg, Luruper Chaussee 149, 22761 Hamburg, Germany}
\affiliation{Max Planck Institute for the Structure and Dynamics of Matter, Luruper Chaussee 149, 22761 Hamburg, Germany}

\author{Dieter Jaksch}
\affiliation{Institute for Quantum Physics, University of Hamburg, Luruper Chaussee 149, 22761 Hamburg, Germany}
\affiliation{The Hamburg Centre for Ultrafast Imaging, University of Hamburg, Luruper Chaussee 149, 22761 Hamburg, Germany}
\affiliation{Clarendon Laboratory, University of Oxford, Parks Road, Oxford OX1 3PU, United Kingdom} 

\begin{abstract}

As quantum hardware enters the fault-tolerant regime, algorithms are beginning to be executed on encoded qubits, yet scalable certification of how well they run remains an open problem. Here we introduce \emph{magic-free certification}, a certification protocol for universal fault-tolerant circuits that consumes no magic, the resource that generally places them beyond efficient classical simulation and is by far the most expensive element of a fault-tolerant computation. The protocol is code- and decoder-agnostic, runs interleaved with the target circuit on the same device, and has a sample complexity that grows linearly with the number of logical qubits and non-Clifford gates while remaining decoupled from the logical error rate itself. Given a predefined error tolerance $p_L^*$, the protocol determines at a chosen confidence level whether the device executes the target circuit at a logical error rate below $p_L^*$. The method relies on \emph{trap circuits}, copies of the target with every magic state replaced by a stabilizer state, whose outcomes are deterministic and which, being indistinguishable from the target to stochastic Pauli noise, fail at the same rate. We illustrate the protocol in Steane-code simulations of a magic state injection gadget and a quantum Fourier transform.

\end{abstract}

\maketitle

\section{Introduction}

Quantum computing is transitioning from the era of noisy intermediate-scale quantum (NISQ) devices \cite{Preskill2018} towards the era of fault
tolerance, in which quantum information is encoded in logical qubits and
processed by error-corrected operations \cite{Shor1995, Steane1996,KnillLaflamme1997,Aharonov1997, Knill1998, Preskill1998}. Below the fault-tolerance threshold, quantum error correction (QEC) can suppress logical error rates exponentially with the code distance. This is no longer a theoretical prospect, since small encoded logical circuits, and even some fault-tolerant primitives, have already been demonstrated across platforms \cite{Anderson2021, Postler2022, Ye2023, Bluvstein2023, Gupta2024,mayer2024benchmarkinglogicalthreequbitquantum, 
GoogleQuantumAI2024,Lacroix2024,
SalesRodriguez2025, CodeSwitching2025, Dasu2025, Bluvstein2026, Zhang2026}. Fault-tolerant circuits are thus beginning to be executed, with their performance so far assessed by methods such as tomographic reconstruction and fidelity estimation, which are viable at current system sizes but will become inefficient as fault-tolerant devices scale.

Ensuring that a quantum device operates correctly, in terms of the accuracy of its output, is the task known as \emph{certification}, or sometimes \emph{verification} \cite{eisert2020, KlieschRoth2021,Carrasco2021perspectives}. For a logical circuit, that accuracy can be quantified by its average logical error rate. The difficulty in such a task is that a reference is needed, something known independently of the device against which its output can be compared. For stabilizer circuits the reference comes at essentially no cost. A stabilizer circuit is one built entirely from Clifford gates and acting on stabilizer states, the states completely specified by a set of commuting Pauli operators that leave them invariant. Such circuits can be simulated efficiently on a classical computer \cite{Gottesman1998, Aaronson2004}, so their ideal output can be computed before the experiment is run and compared directly with what the device returns. However, for universal circuits, which are the ones capable of a computational advantage, no efficient classical method for computing the ideal output is known in general. Universality requires gates outside the Clifford group, and it is precisely these non-Clifford elements that make the classical cost of computing the ideal output grow exponentially.

The lack of an efficiently computable ideal output also constrains methods such as quantum state tomography \cite{Hradil2004} and fidelity estimation \cite{Flammia2011, daSilva2011, huang2020}. These apply to arbitrary circuits, but they reconstruct non-deterministic statistics, so that resolving a small error rate against quantum projection noise demands a number of shots that grows as that rate decreases, on top of resources scaling exponentially with system size \cite{Haah2017}. What is needed is therefore a certificate that does not require the ideal output, whose cost does not grow as the device improves, and which scales gracefully with the size of the circuit.

In a fault-tolerant architecture, those same non-Clifford gates are also the most expensive part of the computation. No quantum error-correcting code admits a transversal implementation of a universal gate set
\cite{EastinKnill2009}, so they cannot be applied directly on encoded qubits
and must instead be supplied by consuming a pre-prepared resource state, known as a \emph{magic state}, which teleports the gate onto the data
\cite{Gottesman1999, Knill2004, Bravyi2005, Campbell2017}. In this context, \emph{magic} is the resource that quantifies how far a quantum state is from being a stabilizer state \cite{Veitch2014, Howard2017}. Magic states lie outside the set of states a fault-tolerant stabilizer circuit can produce, and require dedicated distillation or cultivation routines \cite{Bravyi2012,Litinski2019,gidney2024} whose cost dominates the resource budget of the machine. A certification method that consumes magic of its own therefore competes directly with the computation for its most expensive resource. 

Several approaches do not certify the circuit directly, but characterize individual components or the noise affecting the device, from which the behavior of a given circuit must then be inferred. Gate- and layer-level benchmarking, which includes randomized benchmarking \cite{Emerson2005, Knill2008, Magesan2011, Magesan2012IRB, Onorati2019,Proctor2022mirror} and its non-Clifford extensions \cite{CD2015, Cross2016, Harper2017}, cycle benchmarking \cite{Erhard2019}, and error-per-layer characterization \cite{Hurant2026}, extracts a fidelity for a gate or a circuit layer by averaging over random sequences, and it has already been applied at the logical level \cite{Lacroix2024,mayer2024benchmarkinglogicalthreequbitquantum}. All this methods, however, return a component-wise figure of merit. Noise-model learning reconstructs the error channel itself, either from dedicated experiments \cite{FlammiaWallman2020, Harper2020, vandenBerg2023}, or directly from the syndrome data that error correction
already produces \cite{Wagner2022paulichannelscanbe, Wagner2023, Zheng2026,
xiao2026, Pancotti2026}. This last route is remarkably cheap in shots, since it consumes data the device generates anyway, but it yields a model rather than a verdict. Statements about the circuit follow from the model rather than from the data directly, and hold only if the assumed form of the noise is correct.

A different set of approaches draws conclusions about the target circuit performance from circuits that resemble it structurally, but whose output remains classically computable. The idea originates in delegated quantum computing, where a client delegates a computation to a server that may be actively adversarial. There, trap circuits, indistinguishable to the server from the computation itself but with outcomes the client knows in advance, are
inserted among the target executions, and any deviation the server attempts can be caught by the user \cite{Broadbent2009, Fitzsimons2017, Gheorghiu2015, Mahadev2018}, with this construction recently extended to computations that consume magic \cite{Sater2026}. These ideas also apply to hardware that is noisy, rather than deliberately malicious. In accreditation, the target computation is interleaved with trap circuits whose correct output is known in advance, and the observed trap failures are converted into a bound on the error of the target's output, first for NISQ circuits \cite{Ferracin2019, Ferracin2021,Jackson2024} and recently for universal fault-tolerant ones
\cite{Mills_2026}. A separate strategy keeps the idea of resembling circuits but builds them from matchgates, a different classically simulable family that also becomes universal once an additional resource is supplied. Removing that resource leaves a circuit close to the target while being classically simulable \cite{Carrasco2024}. 

None of these approaches, however, returns a verdict against a tolerance the user fixes in advance. Each first estimates a quantity (an error rate, or a bound on it) and compares afterwards, so the measurement effort is set by the precision to which that quantity happens to be resolved rather than by the question being asked. This is what we provide. We introduce \emph{magic-free certification}, a protocol that decides whether a device executes a given fault-tolerant circuit below a logical error tolerance $p_L^*$ fixed in advance, at a confidence level $1-\delta$ the user also chooses.

As in accreditation, the decision rests on trap circuits. These are copies of the target in which every magic state is replaced by a stabilizer state, so that they contain none of the magic that is the most expensive element of a fault-tolerant computation, and reduce to Clifford circuits whose input can be prepared so that their output is deterministic. Because stochastic Pauli noise cannot distinguish a magic state from a stabilizer state, the traps are corrupted at the same rate as the target, and every flipped trap outcome is a logical error detected in a single execution. This is analogous to the \emph{magic-blindness} of Pauli noise,
introduced by Sater and Ollivier \cite{Sater2026} in the adversarial setting described above. Here it follows directly from the noise being stochastic Pauli at every circuit location, since a location-attached channel acts identically on any input state. Counting trap failures bounds the logical error rate of the target, and a single hypothesis test on the count returns the verdict. The result is a certification protocol for universal fault-tolerant circuits that is code- and decoder-agnostic, runs interleaved with the target circuit on the same device, and consumes no non-Clifford resources of its own. Its sample complexity grows linearly with the number of logical qubits and non-Clifford gates, and is independent of the logical error rate itself.

The remainder of this paper is organized as follows. Section
\ref{main_results} presents the protocol, the bound on which it rests, and the statistical test that turns that bound into a certificate. Section \ref{framework} develops the trap construction in detail, states the
assumptions, and proves the bound together with its tightness. Section \ref{protocols} derives the two operational protocols and their sample
complexities. Section \ref{numerics} illustrates both in Steane-code simulations of a magic state injection (MSI) gadget and a quantum Fourier transform (QFT). Section \ref{discussion} summarizes the work, connects
the protocol to current experiments, and outlines extensions.

\section{Main results}\label{main_results}

\begin{figure*}
    \centering
    \includegraphics[width=\textwidth]{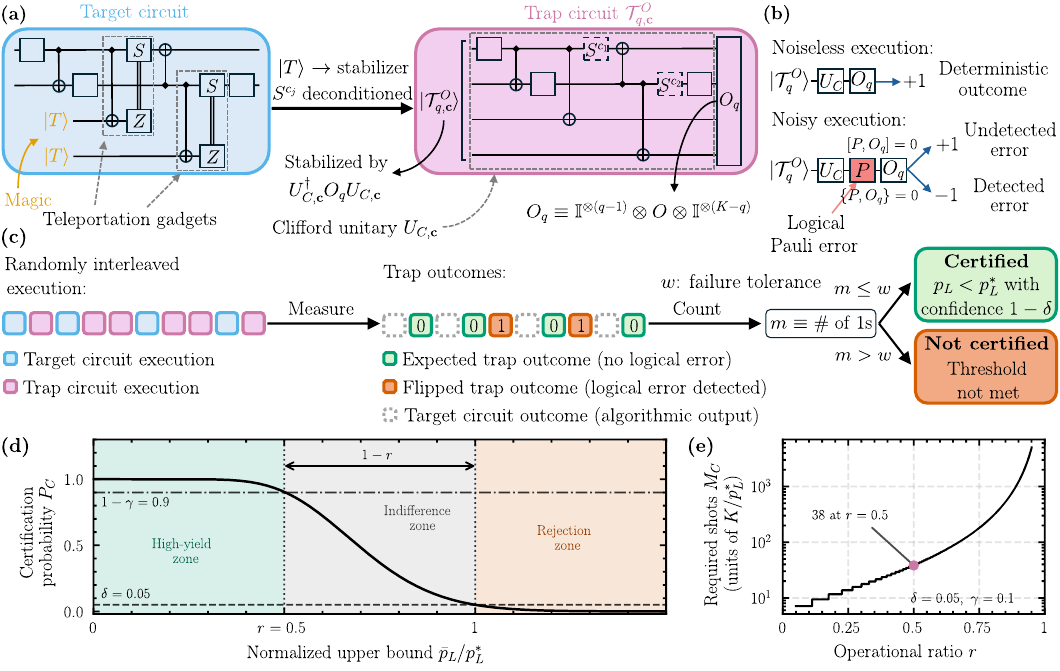}
    \caption{Magic-free certification of a universal fault-tolerant quantum circuit. (a) Trap construction. A target circuit (blue box) built from Clifford and $T$ gates, with the $T$ gates implemented fault-tolerantly via magic state injection. Each $T$ gate is teleported through a gadget consisting of a logical CNOT onto
an ancilla block prepared in $\ket{T}$, a logical $Z$-basis measurement of the ancilla with outcome $a_j$, and a classically controlled correction $S^{a_j}$. The associated trap circuits $\mathcal{T}_{q,\mathbf{c}}^O$ (pink box) substitute every injected magic state with a stabilizer resource state and apply the corrections according to a uniformly sampled correction bitstring $\mathbf{c}$ rather than to the measurement outcomes, leaving a Clifford unitary $U_{C,\mathbf{c}}$. The joint input is prepared in the $+1$ eigenspace of $U_{C,\mathbf{c}}^\dagger O_qU_{C,\mathbf{c}}$, where $O_q\equiv\mathbb{I}^{\otimes(q-1)}\otimes O\otimes\mathbb{I}^{\otimes(K-q)}$ applies the Pauli $O\in\{X,Y,Z\}$ to logical qubit $q\in\{1,\dots,K\}$, and $K$ counts the data qubits and the ancilla blocks consumed by the injections. (b) Logical error detection, with $\mathbf{c}$ omitted since the mechanism is the same for every $\mathbf{c}$. Without noise, the measurement of $O_q$ returns $+1$ deterministically. A residual logical Pauli error $P$ flips it to $-1$ if it anticommutes with $O_q$, so every flip signals a logical error,
while errors commuting with $O_q$ go undetected by that trap. (c) Certification protocol. Trap and target circuits are executed in random order on the same device, so that both sample the same time-averaged noise. Trap outcomes are recorded as 0 if the deterministic value is recovered and 1 if it is flipped, and the number $m$ of 1s over $M_C$ shots distributed among the $3K$ traps is compared against a failure tolerance $w$. The circuit is certified to satisfy $p_L<p_L^*$, for a predefined tolerance $p_L^*$, with confidence $1-\delta$ if $m\leq w$, and is left uncertified otherwise. (d) Certification probability $P_C$ against the normalized upper bound $\bar p_L/p_L^*$, for operational ratio $r=0.5$. The tolerance $w$ and the shot budget $M_C$ are calibrated so that a device in the high-yield zone ($\bar p_L\leq rp_L^*$) is certified with probability at least $1-\gamma=0.9$ (dash-dotted line), and with probability at most $\delta=0.05$ (dashed line) if it sits in the rejection zone ($\bar p_L\geq p_L^*$), while in the indifference zone between them, of width $1-r$, neither guarantee applies. The curve depends only on $\delta$, $\gamma$ and $r$, and is the same for every circuit.
(e) Shot budget $M_C$, in units of $K/p_L^*$, required by the test at operational ratio $r$, for $\delta=0.05$ and $\gamma=0.1$. Stricter values of either raise it further. The dot marks the value $r=0.5$ used in (d).}
    \label{main_figure}
\end{figure*}

We now describe the magic-free certification protocol in full, including how trap circuits are built, how their combined failure rates bound the logical error rate of the target, and how a single hypothesis test converts them into a certificate at a prescribed error tolerance and confidence level.

Figure \ref{main_figure}(a) shows the construction of trap circuits. Throughout, we take the non-Clifford elements of the target circuit to be single-qubit $\pi/4$-rotations about the $z$-axis, $R_z(\pi/4)$, commonly known as $T$ gates. The construction extends to other non-Clifford gates, as discussed in Sec. \ref{traps_subsection}. The target circuit, shown on the left of Fig. \ref{main_figure}(a), implements each $T$ gate by teleporting it through a gadget (dashed boxes) that consumes a magic state $\ket{T}=T\ket{+}= (|0\rangle + e^{i\pi/4}|1\rangle)/\sqrt{2}$ prepared on an ancilla block, measures the ancilla block, and applies a correction $S^{a_j}$ conditioned on the outcome $a_j$, which can be either 0 or 1 with equal probability, independently of the data it acts on. We denote by $K$ the total number of logical qubits of the circuit, counting both the data blocks and the blocks consumed by the injections, and label them by $q\in\{1,\dots,K\}$. In the associated trap circuits, shown on the right of Fig. \ref{main_figure}(a), every magic state is replaced by a stabilizer resource state, and the correction gates (dashed $S^{c_j}$ boxes) are applied unconditionally, according to a bitstring $\mathbf{c}$ drawn uniformly at random rather than according to the measurement record. Sampling $\mathbf{c}$ uniformly therefore reproduces the distribution of corrections the target itself realizes, while removing the classical feedforward. If the corrections are tracked in software rather than applied physically, sampling of $\mathbf{c}$ is not required, and the correction slots are removed. In any case, what remains is a Clifford unitary $U_{C,\mathbf{c}}$ (big dashed box), built from the same gates, at the same time slots, as the target, and including its full schedule of syndrome-extraction rounds.

Each trap measures a logical Pauli $O\in\{X,Y,Z\}$ on one logical qubit $q$ at the end of the circuit, that is, the operator $O_q\equiv\mathbb{I}^{\otimes(q-1)}\otimes O\otimes\mathbb{I}^{\otimes(K-q)}$. Given the target circuit, a trap is fully specified by $O$, $q$ and $\mathbf{c}$, and we therefore denote it by $\mathcal{T}^O_{q,\mathbf{c}}$. Because trap circuits are Clifford, the input states that make their output deterministic can be computed efficiently. Preparing the joint input, labeled $\ket{\mathcal{T}^O_{q,\mathbf{c}}}$ in Fig.~\ref{main_figure}(a), in a $+1$ eigenstate of $U_{C,\mathbf{c}}^\dagger O_qU_{C,\mathbf{c}}$ ensures that the measurement of $O_q$ at the end of the circuit returns a fixed $+1$ value, which we record as the bit 0 (Fig.~\ref{main_figure}(b), top). Since the outcome is fixed in the absence of noise, a flip in the expected outcome can only be caused by a physical fault that has escaped error correction and left a logical error on qubit $q$ (Fig.~\ref{main_figure}(b), bottom). We then define a trap failure as an execution in which the deterministic value is flipped. Every trap failure is therefore a detected logical error, identified in a single execution, without estimating any probability and without resolving a small signal against quantum projection noise.

The validity of trap circuits as proxies for the target, the premise of the whole protocol, rests on noise being stochastic Pauli. At every location of the target and trap circuits (gates, idle periods and measurements), such noise draws an error from a fixed distribution and applies it to whatever state is passing through. We call the set of physical errors occurring during one execution a \emph{fault configuration} $F$. Since the channel is a fixed map, the distribution from which errors are drawn does not depend on the state. Thus, the fault configurations of the target circuit, where $\ket{T}$ states are injected, are statistically identical to those occurring in trap circuits when a stabilizer state is injected instead. Substituting magic states by stabilizer states therefore leaves the fault statistics of the circuit untouched while removing its non-Clifford content. Three further assumptions complete the model: the injected magic states carry stochastic errors no more frequent than those of the stabilizer states replacing them, the noise on preparing an encoded qubit does not depend on which basis state is prepared, and the decoder is held fixed throughout the experiment.

Under the noise model above, three properties of the trap construction turn the trap outcomes into a statement about the target circuit. First, Pauli errors propagate through Clifford gates by conjugation alone, an operation that is deterministic and independent of the state acted upon, so a given fault configuration $F$ that is not corrected during QEC rounds produces the same residual logical Pauli error $P(F)$ in the target and in any of its trap variants. Second, the noise at a circuit location depends only on the operation executed there, and target and traps execute identical operations at every location except where the resource states are prepared. Since stochastic Pauli noise acts identically on a magic state and on a stabilizer state, those locations carry the same noise as well, and the fault statistics of the two circuits agree location by location. Third, a trap sees only those errors that anticommute with the Pauli it measures (Fig.~\ref{main_figure}(b)). Having prepared the input as a $+1$ eigenstate, the final measurement returns $-1$, recorded as a bit 1, precisely when the residual logical error on qubit $q$ anticommutes with $O$. One trap alone is therefore blind to a part of the error, since a trap measuring $Z$, for instance, cannot detect a residual $Z$. Each non-trivial Pauli anticommutes with
exactly two of $X$, $Y$ and $Z$, so executing three traps on every logical
qubit, one for each Pauli, registers each of that qubit's errors
twice, and half the sum of their failure rates is the probability that qubit
$q$ carries a logical error at all. The computation is corrupted if any of
the $K$ logical qubits is corrupted and, by the union bound, that probability (the logical error rate $p_L$ of the target circuit) is at most the sum of the individual ones, which yields
\begin{equation}\label{overview_bound}
    p_L\leq\dfrac{1}{2}\sum_{q=1}^K\sum_{O\in\{X,Y,Z\}}p\left(\mathcal{T}_{q}^O\right)\equiv\bar p_L,
\end{equation}
with $p\left(\mathcal{T}_{q}^O\right)$ denoting
the logical failure rate of the trap measuring $O$ on logical qubit $q$,
averaged over the correction bitstrings $\mathbf{c}$. The left-hand side is a property of a circuit whose output is neither deterministic nor efficiently predictable, whereas the right-hand side is assembled entirely from executions of stabilizer circuits on the device that runs it.

The protocol connects the trap failure counts, which are what an experiment records, to the failure rates entering Eq.~\eqref{overview_bound} through a hypothesis test, sketched in Fig.~\ref{main_figure}(c). Trap and target circuits are executed in random order on the same device, so that both sample the same time-averaged noise and the certificate refers to the conditions under which the target actually runs. The total number $m$ of flipped trap outcomes over a budget of $M_C$ shots is compared against a failure tolerance $w$. Observing $m\leq w$ certifies $p_L< p_L^*$ at confidence $1-\delta$, with $\delta$ the false positive rate, whereas observing more leaves the circuit uncertified. The tolerance $w$ and the shot budget $M_C$ are calibrated so that a device operating at $\bar p_L=p_L^*$ is wrongly certified with probability at most $\delta$, while one operating comfortably below $p_L^*$, at $\bar p_L\leq r\,p_L^*$, for an operational ratio $0<r<1$ that sets how far below the tolerance a device must operate to be reliably certified, is certified with probability at least $1-\gamma$, with $\gamma$ the false negative rate.

Figure~\ref{main_figure}(d) shows the resulting certification probability for
$r=0.5$, plotted against the normalized upper bound $\bar p_L/p_L^*$. The
curve crosses the significance level $\delta$ at the tolerance $p_L^*$ and
meets the yield level $1-\gamma$ at the operational target $rp_L^*$, leaving
between them an indifference zone, in which neither guarantee applies and the
outcome may go either way, and whose width $1-r$ is set by the choice of $r$.
The cost of the test is set by how many failures must be observed to separate
the two hypotheses, not by how precisely a small rate must be resolved, and
this is what decouples it from the quality of the device. The required
certification shot budget scales as
\begin{equation}
    M_C\sim\mathcal{O}\left(\dfrac{K}{p_L^*}\right)
\end{equation}
linear in the logical footprint $K$ and inversely proportional to the
predefined tolerance $p_L^*$, with no dependence on the hardware error rate.
Its prefactor depends only on $\delta$, $\gamma$ and $r$, and grows as the
indifference zone narrows, as shown in Fig.~\ref{main_figure}(e) and derived in Sec.~\ref{cert_protococol}. A sharper
test, with $r$ closer to one, must observe more failures to separate the two
hypotheses, and therefore requires more shots. To give a sense of scale,
consider the compilation of 2048-bit RSA factorization in
Ref.~\cite{gidney2025factor2048bitrsa}, which uses $6.5\cdot10^9$ non-Clifford
gates on 1399 logical qubits. Certifying that computation as a single object
might not be practical, since the shot budget grows linearly with the logical
footprint. Its arithmetic is built, however, from a small set of repeated
primitives, modular additions and table lookups acting on registers of a few
tens of qubits. Certifying one such primitive to a tolerance $p_L^*=10^{-3}$,
at $(\delta,\gamma,r)=(0.05,0.1,0.5)$, requires $M_C\sim10^6\text{--}10^7$
shots.

\section{Magic-free trap framework}\label{framework}

This section develops the construction summarized in Sec.~\ref{main_results}. It defines the target and trap circuits precisely by detailing the construction of the latter, states the assumptions under which they are equivalent in terms of the errors they suffer, and proves the bound of Eq.~\eqref{overview_bound} together with its tightness.

Throughout this work we operate at the logical level: circuits act on encoded qubits, gates are fault-tolerant logical operations on an underlying physical layout, errors are the residual logical Pauli operators left after decoding, and noise channels are the effective channels induced on the logical qubits by the underlying physical noise. Unless stated otherwise, all quantities refer to this level.

\subsection{Target circuits and trap construction}\label{traps_subsection}

We consider target circuits compiled into the $\text{Clifford}+T$ gate set. $T$ gates are then implemented by injection of the magic state $\ket{T}$ via the standard teleportation gadget, as shown in Fig.~\ref{main_figure}(a). The gadget consists of a CNOT with the data block as control and the ancilla block as target, followed by a $Z$-basis measurement of the ancilla block with outcome $a_j\in\{0,1\}$. When $a_j=1$, the gadget applies $T^\dagger=S^\dagger T$, where $S=R_z(\pi/2)$, that is, the desired $T$ followed by a byproduct $S^\dagger$. A classically controlled correction $S^{a_j}$ compensates the byproduct of the teleportation and restores $T$. A target circuit acting on $k_D$ logical data qubits and containing $k_T$ injected $T$ gates requires one injection ancilla per $T$ gate, so that the total logical footprint of the computation is $K=k_D+k_T$. We stress that $k_T$ counts the injected gates, not the distinct physical blocks, which may be reused across the circuit.

Trap circuits $\mathcal{T}_{q,\mathbf{c}}^O$, with $\mathbf{c}\in\{0,1\}^{k_T}$, are obtained from the target by (i) substituting every injected magic state of the target circuit by a stabilizer resource state and applying the correction $S^{c_j}$ unconditionally at the $j$th correction slot, (ii) preparing the joint input (comprising both data and injection ancilla blocks) in a stabilizer state within the $+1$ eigenspace of the trap operator $U_{C,\mathbf{c}}^\dagger O_qU_{C,\mathbf{c}}$ and (iii) measuring $O_q$ at the end of qubit $q$'s live window, i.e. the interval between its preparation and its measurement in the target circuit. Here $U_{C,\mathbf{c}}$ denotes the Clifford unitary obtained from the target circuit after the magic states are replaced by stabilizer resource states and the correction bitstring $\mathbf{c}$ is fixed. This construction follows Ref.~\cite{Sater2026}, extended here from $Z$-type traps to the full logical Pauli basis. Since each block is measured at the end of its live window, injection ancillas are measured at the same circuit time as in the target, so that the trap does not accumulate idling locations absent from it, whereas data blocks are measured at the end of the circuit.

The bitstring $\mathbf{c}$ enumerates the possible injection outcomes, and hence the correction bitstrings that the target may realize. In every trap shot the correction bitstring \textbf{c} is drawn uniformly at random and the corrections $S^{c_j}$ applied unconditionally, so that the trap ensemble reproduces the target's distribution over correction bitstrings, which requires $k_T$ random bits per shot. For small $k_T$, however, the sampling may equivalently be derandomized by executing each of the $2^{k_T}$ patterns an equal number of times, as done in Sec.~\ref{numerics}. As shown in Appendix \ref{lemma2}, the target's correction bitstring is itself uniformly distributed and independent of the faults. If corrections are tracked in software rather than applied physically, the correction slots are absent from both target and trap, and no sampling of $\mathbf{c}$ is required.

By construction, target and trap circuits share the same circuit location set: identical gate and idle locations, identical measurement times, and identical correction slots. They differ only in the injected resource states and in the prepared input state on the data blocks. The equivalence of the location set extends to error correction, since every logical block undergoes the same schedule of syndrome extraction rounds in the traps as in the target, throughout its live window, irrespective of which qubit the trap is targeting. Restricting error correction to the targeted block would break the correspondence, since faults occurring during syndrome extraction on the remaining blocks propagate onto block $q$ in the target but would never be sampled by its traps. Conversely, blocks whose live window has closed undergo no further rounds, mirroring the target.

Although we have specified the construction for the $\text{Clifford}+T$ gate set, it applies to any non-Clifford gate whose injection gadget leaves a Clifford correction, since fixing that correction according to $\mathbf{c}$, or tracking it in software, then keeps $U_{C,\mathbf{c}}$ Clifford and the trap input state efficiently computable. This holds for every gate in the third level of the Clifford hierarchy \cite{Gottesman1999}, which includes the $T$, controlled-$S$ and controlled-controlled-$Z$ gates. Gadgets that may consume several ancilla blocks, as for the latter two, simply enlarge the logical footprint and the length of $\mathbb{c}$ accordingly.

\subsection{Assumptions}\label{assumptions_sec}

Sharing the location set is not by itself enough for target and traps to suffer the same errors with the same probabilities. We further rely on a set of assumptions concerning the physical noise and the preparation of input states. We state them below as A1--A4, each followed by a brief discussion of its role and of when it is expected to hold.

\begin{enumerate}[
  label=A\arabic*.,
  ref=A\arabic*,
  leftmargin=*,
  align=left,
  itemsep=0.4em,
  topsep=0.4em,
  series=assumptionslist
]
\item Noise is Markovian and is described, at each circuit location, by a stochastic Pauli channel.

\end{enumerate}

Stochastic Pauli noise can be enforced in practice by randomized compiling
\cite{Wallman2016,Hashim2021}, in which random Pauli operations are inserted
around circuit layers so that the circuit logic is unchanged and generic noise
is twirled into a stochastic Pauli channel at each location. Although
introduced for physical circuits, randomized compiling has also been extended
to fault-tolerant encoded circuits \cite{Beale2023,Mills_2026}. Moreover, QEC itself also pushes in this direction, since repeated syndrome
extraction tends to suppress the coherent part of the noise at the logical
level \cite{Beale2018,IversonPreskill2020}. We take this noise structure as
given and do not analyze the compilation itself. Note that state-independence of the location-attached channels is not an additional assumption but a consequence of linearity, since a location-attached channel acts identically on any input state. Nevertheless, Appendix~\ref{coherent_appendix} quantifies how relaxing A1 affects the bound in Eq.~\eqref{overview_bound}, allowing each channel to retain a small coherent component.

This noise model excludes leakage and erasures, which form a significant part
of the error budget of several platforms, but these can be brought within it
once they are detected and converted into Pauli-type noise
\cite{aliferis2006,Fowler2013,suchara2014}. In neutral-atom processors, atom loss and
leakage out of the computational space can be converted into heralded erasures
\cite{wu2022erasure,Ma2023,Scholl2023,Chow2024} and handled by decoders that
exploit the loss locations \cite{Perrin2025,Baranes2026}. In trapped ions,
leakage out of the hyperfine qubit manifold can be pumped back optically
\cite{Hayes2020} or converted into erasures via metastable states
\cite{Kang2023}, and qubit loss can be detected and corrected deterministically
\cite{Stricker2020}. In superconducting circuits, leakage to higher transmon
levels is removed by multilevel reset and leakage-reduction units
\cite{McEwen2021,Miao2023,Battistel2021}, while dual-rail encodings convert
the dominant decay channel into erasures \cite{Kubica2023}. Once
converted and reset, an erasure acts on the logical state as a random Pauli
error at a known location, and the heralds enter the decoder as an additional
input.

Correlations between executions need to be considered separately. Slow noise drifts are harmless, since random interleaving makes target and traps sample
the same time-averaged noise, whereas correlations on the scale of consecutive
executions may have a negative impact on the protocol. Residual leakage or heating that persists from one shot
into the next can make the noise in a shot depend on the circuit executed
before it, breaking the equality of fault statistics. High-energy impacts cause errors across chip arrays in superconducting
processors \cite{Wilen2021,McEwen2022}, correlating failures across nearby
shots, so that the failure count is no longer a sum of independent trials and
the confidence level of the test is no longer guaranteed. The former are
addressed by resetting every qubit between executions, while the latter by hardware
mitigation such as gap engineering \cite{McEwen2024}.

\begin{enumerate}[
resume=assumptionslist,
  label=A\arabic*.,
  ref=A\arabic*,
  leftmargin=*,
  align=left,
  itemsep=0.4em,
  topsep=0.4em
]
\item The injected magic states of the target circuit carry errors of stochastic Pauli type, with a logical error rate not exceeding that of the substituted stabilizer resource states in the trap circuits.

\end{enumerate}

\noindent Together with the preparation of the input states, this is the only difference between the physical instructions of target and traps. If the injected magic states were noisier than the stabilizer states replacing them, the target would suffer faults that its traps never sample, and the bound would underestimate its error.

In fault-tolerant settings, magic states are supplied by a specialized
preparation stage, such as distillation \cite{Bravyi2005,Bravyi2012,Litinski2019}
or cultivation \cite{gidney2024}. The logical-level preparation and distillation of magic states have now been demonstrated with superconducting circuits \cite{Ye2023,Gupta2024}, neutral atoms \cite{SalesRodriguez2025} and trapped ions \cite{CodeSwitching2025,Dasu2025}. The certification of this stage
lies outside the scope of the protocol and is assumed to be performed
independently with dedicated benchmarks of logical magic states
\cite{Lee2025magicbenchmark}. If the magic states are shown to be noisier than the stabilizer states they replace, the bound of Eq.~\eqref{overview_bound} degrades only additively, by at most the excess error per injection times the number of injections. The mismatch can
also be absorbed by deliberately degrading the stabilizer resource states of
the traps until the inequality holds, at the cost of a slightly looser bound.

Stochasticity of the preparation error can be enforced by twirling the
preparation stage, a standard step in the magic-state literature
\cite{Bravyi2005,Bravyi2012} for which fault-tolerant constructions on encoded
states have also been developed \cite{Piveteau_2021,Beale2023,Mills_2026}.

\begin{enumerate}[
  resume=assumptionslist,
  label=A\arabic*.,
  ref=A\arabic*,
  leftmargin=*,
  align=left,
  itemsep=0.4em,
  topsep=0.4em
]
\item The noise channel attached to the preparation of an encoded logical state is independent of which stabilizer basis state is prepared.
\end{enumerate}

\noindent This is required because target and traps prepare different stabilizer inputs on the data blocks. Although A1 attaches a stochastic Pauli channel to each operation, the preparations of different basis states are different operations whose channels might not coincide. A3 holds by construction when every basis state is obtained from one common preparation followed by a Clifford basis change, since the preparation
is then the same operation in every case and the basis change is a separate
gate of the circuit, as is possible in codes admitting transversal Clifford
gates \cite{Steane1996,bombin2006}. However, when different basis states are prepared by different circuits, A3 becomes a hypothesis about the device to be verified for the used preparation circuits. If it fails, target and traps no longer share the same fault statistics at the preparation locations, and the bound of Eq.~\eqref{overview_bound} requires an additive correction.

\begin{enumerate}[
  resume=assumptionslist,
  label=A\arabic*.,
  ref=A\arabic*,
  leftmargin=*,
  align=left,
  itemsep=0.4em,
  topsep=0.4em
]
\item The decoder is fixed throughout the experiment.
\end{enumerate}

\noindent Errors considered here are the residual logical errors left after decoding, and those occurring in trap circuits are therefore the same occurring in the target only if both are decoded by the same map. The decoder may process any number of syndrome rounds and use side information, such as erasure heralds, provided the map from syndrome data to correction is the same in every execution. A4 therefore excludes decoders that change while the certification is running, which would be the case if decoders that are tuned to device data \cite{Bausch2024} were retrained or
recalibrated between shots. All error rates are thus defined relative to that decoder, and
the protocol certifies the device together with its decoding pipeline rather
than the device alone.

Assumptions A1--A4 are the only hypotheses entering the bound derived below, as no further requirement is imposed on the code, the circuit, or the noise.

\subsection{Logical error rate upper bound}

Throughout this section we compare the execution of the target circuit with those of its trap variants.  Under A1, the probability of each fault configuration $F$, which was the set of physical errors occurring during a single execution, is set by the noise model alone.

Each fault configuration propagates through the circuit and is partially corrected by the decoder, resulting, after decoding, in a residual logical Pauli error
\begin{equation}\label{logical_pauli_error}
    P(F)=\bigotimes_{q=1}^KP_q,
\end{equation}
where $P_q\in\{\mathbb{I},X,Y,Z\}$ is the error carried by logical qubit $q$ at the moment it is measured. The map $F\mapsto P(F)$ also depends on the correction bitstring $\mathbf{c}$ and on the decoder, which we leave implicit. We write $E_q=\{F:P_q\neq\mathbb{I}\}$ for the fault configurations leaving a logical error on qubit $q$, and $E_\text{out}$ for those corrupting the output of the target circuit. If $P(F)=\mathbb{I}$, no logical qubit carries an error when it is measured, and the target behaves as its ideal counterpart, so that
\begin{equation}
    E_\text{out} \subseteq \{F : P(F) \neq \mathbb{I}\}.
\end{equation}
Not every logical error corrupts the target circuit output. For instance, a logical $Z$ error on an injection ancilla at the moment of its $Z$-basis measurement leaves the outcome unchanged. Counting all logical errors therefore overestimates $p_L$, a gap that is analyzed in Sc.~\ref{bound_tightness}.

The bound then follows from three lemmas: a given fault configuration produces the same logical error in the target and in every trap (Lemma~\ref{same_logical_error}), target
and traps encounter fault configurations with the same statistics (Lemma~\ref{same_fault_stats}), and the three traps of a logical qubit
together register every logical error on it exactly twice (Lemma~\ref{error_twice}).

Lemma~\ref{same_logical_error} formalizes the first property anounced in Sec.~\ref{main_results}: once the correction bitstring is fixed, a fault
leaves the same logical error whatever the state it acts on.
\begin{lemma}\label{same_logical_error}
    For a fixed fault configuration $F$, correction bitstring $\mathbf{c}$, and decoder, the resulting logical error $P(F)$ is the same in the target circuit and in every trap circuit, and is independent of the logical input state.
\end{lemma}
\begin{proof}
    Once $\mathbf{c}$ is fixed, the corrections $S^{c_j}$ are fixed gates. In the target, $\mathbf{c}$ is the record of injection outcomes, which fixes its corrections in the same way. Target and traps then execute the same Clifford gates at the same locations and differ only in the states they process. Pauli faults propagate through Clifford gates by conjugation, which does not depend on the state acted upon, so the syndromes, and hence the decoder's corrections, depend only on $F$ and $\mathbf{c}$. The map
    $F\mapsto P(F)$ is therefore the same in all executions.
\end{proof}

Lemma~\ref{same_logical_error} fixes the map from faults to logical errors. Target and traps could still fail at different rates if they encountered faults, correction bitstrings or decoders with different statistics, which the next lemma excludes everywhere except at the resource-preparation slots.
\begin{lemma}\label{same_fault_stats}
    Under A1, A3 and A4, and with $\mathbf{c}$ sampled uniformly in the traps, target and trap executions encounter correction bitstrings with the same distribution, are decoded by the same map, and sample fault configurations with the same statistics at every location except the resource-preparation slots.
\end{lemma}
\begin{proof}[Proof sketch]
    Given the correction bitstring, every location other than resource preparation carries the same physical instruction in target and trap: gate layers, idle periods, and measurements are identical by construction, and the data-block preparations differ only in the prepared basis, which carries the same noise by A3. Under A1, the channel attached to a location is fixed independently of the state passing through it, and the fault probabilities agree location by location. The correction bitstrings agree because the target's recorded injection outcomes are uniformly distributed and independent of the faults, so that they follow the same law as the bitstrings \textbf{c} sampled in the traps (Appendix \ref{lemma2}). Finally, syndromes are independent of the logical input state and the decoder is common to all executions by A4, so the decoder responses agree as well.
\end{proof}

The two executions still differ at the resource-preparation slots, where the target injects magic states and the traps stabilizer states. However, by A2 the injected magic states are not noisier than the stabilizer states they replace.

Lemmas \ref{same_logical_error} and \ref{same_fault_stats} establish the trap failure rates as faithful probes of the target's logical errors, but a single trap reports only whether the error anticommutes with the observable it measures, and is blind to the rest. Therefore, obtaining the full marginals requires the complete logical Pauli basis.
\begin{lemma}\label{error_twice}
    For every logical qubit $q$,
    \begin{equation}\label{erro_tw_eq}
        p\left(\mathcal{T}_q^X\right)+p\left(\mathcal{T}_q^Y\right)+p\left(\mathcal{T}_q^Z\right)=2p\left(E_q\right),
    \end{equation}
    where $p(\mathcal{T}_q^O)$ are the trap failure probabilities, i.e. the probabilities that the deterministic output of each trap circuit is flipped, and $p(E_q)\equiv\Pr(P_q\neq\mathbb{I})$ are the single-qubit marginals.
\end{lemma}
\begin{proof}
    When executing the trap $\mathcal{T}_q^O$, the deterministic output of the observable $O_q$ flips when the $q$th component of the resulting logical error anticommutes with $O_q$, and errors on the remaining blocks commute with $O_q$ and leave it unaffected. If $P_q=\mathbb{I}$ no trap of qubit $q$ flips, whereas if $P_q\in\{X,Y,Z\}$ it anticommutes with the two remaining Paulis. For every fault configuration we therefore have the identity
    \begin{equation}
        \sum_{O \in \{X,Y,Z\}} \mathbf{1}\left(\mathcal{T}^O_q \text{ flips}\right) = 2\cdot\mathbf{1}(P_q \neq \mathbb{I}),
    \end{equation}
    where $\mathbf{1}\left(\cdot\right)$ is the indicator function. By Lemmas \ref{same_logical_error} and \ref{same_fault_stats} the three traps of qubit $q$ sample fault configurations with the same statistics and map them to the same logical errors, so both sides of the above identity may be averaged over a common distribution. The left-hand side yields $\sum_Op(\mathcal{T}_q^O)$, and the right-hand side $2p(E_q)$.
\end{proof}
Equation \eqref{erro_tw_eq} is the bridge between what the experiment
records and what the bound requires. The marginals $p(E_q)$ are properties
of the target circuit and are not directly accessible, since they refer to
logical errors on a state whose output is neither deterministic nor
efficiently predictable. The trap failure probabilities $p(\mathcal{T}_q^O)$, in
contrast, are estimated by counting flipped outcomes over repeated
executions of stabilizer circuits on the device itself. Lemma
\ref{error_twice} converts the latter into the former.

Combining the three lemmas yields the central result of this work: an upper
bound on the logical error rate of a universal fault-tolerant circuit,
expressed entirely in terms of quantities measured on stabilizer circuits.
\begin{theorem}\label{main_theorem}
    Under assumptions A1--A4 and the trap construction from Sec.~\ref{traps_subsection}, the logical error rate of the target circuit $p_L\equiv\Pr(E_\mathrm{out})$ obeys
    \begin{equation}\label{bound}
        p_L\leq\Pr(P\neq\mathbb{I})\leq\sum_{q=1}^Kp(E_q)=\dfrac{1}{2}\sum_{i=1}^{3K}p(\mathcal{T}_i)\equiv\bar p_L,
    \end{equation}
    where $\{\mathcal{T}_i\}_{i=1}^{3K}$ indexes the complete set of trap circuits $\{\mathcal{T}_q^X,\mathcal{T}_q^Y,\mathcal{T}_q^Z\}_{q=1}^K$.
\end{theorem}
\begin{proof}
    By Lemmas \ref{same_logical_error} and \ref{same_fault_stats}, target and traps share both the map from fault configurations to logical errors and the statistics of those configurations, so every probability appearing above may be evaluated on either.

    The first inequality follows from the inclusion $E_\text{out}\subseteq\{F:P(F)\neq\mathbb{I}\}$, since any fault that corrupts the output of the target circuit necessarily leaves a non-trivial logical Pauli error.

    For the second inequality, a non-trivial logical Pauli acts non-trivially on at least one logical qubit, so that $\{F:P(F)\neq\mathbb{I}\}=\bigcup_{q=1}^K E_q$, and Boole's inequality gives $\Pr(P\neq\mathbb{I})\leq\sum_qp(E_q)$.

    The final equality is Lemma \ref{error_twice} summed over the $K$ logical qubits.
\end{proof}

The theorem turns the certification problem into a measurement of $3K$ deterministic binary outcomes. The price of this reduction is that the relation is an inequality rather than an identity, with the two inequalities having distinct origins. Whereas the first inequality reflects the errors the target circuit is immune to but the traps still register, the second one overcounts single faults corrupting several logical qubits at once as a consequence of the union bound. The quantification of these two sources of looseness is precisely what the following subsection addresses.

\subsection{Tightness of the bound}\label{bound_tightness}

Theorem \ref{main_theorem} bounds the logical error rate without any prior knowledge of the structure of the underlying Pauli channel. That is, neither the individual fault rates, nor their distribution among the circuit locations, nor the relative weights of the different Pauli error types enter the derivation. We now make the two inequalities of Theorem \ref{main_theorem} quantitative, showing that each contributes a multiplicative factor set by the structure of the circuit.

For a fault configuration $F$ producing a logical error as in Eq.~\eqref{logical_pauli_error}, we define its weight as
\begin{equation}
    W(F)=\left|\{q:P_q\neq\mathbb{I}\}\right|=\sum_{q=1}^K\mathbf{1}(F\in E_q),
\end{equation}
that is, the number of logical qubits on which the error acts non-trivially. Using the linearity of the expected value gives
\begin{equation}\label{mean_weight}
    \langle W\rangle=\sum_{q=1}^K\langle\mathbf{1}(F\in E_q)\rangle=\sum_{q=1}^Kp(E_q)=\bar p_L,
\end{equation}
which shows that the upper bound is precisely the mean number of logical qubits corrupted per execution. Being a mean count rather than a probability, $\bar p_L$ can in principle exceed one, and is thus informative only in the fault-tolerant regime, where logical errors are rare and $\bar p_L\ll1$. Conditioning on whether an error occurred at all, with $\Pr(W\geq1)=\Pr(P\neq\mathbb{I})$ and $W=0$ otherwise,
\begin{equation}\label{weight_decomp}
    \langle W\rangle=\langle W|W\geq 1\rangle\cdot\Pr(P\neq\mathbb{I})=\bar W\cdot\Pr(P\neq\mathbb{I}),
\end{equation}
with $\bar W=\langle W|W\geq 1\rangle\geq1$ the mean weight conditioned on a logical error having occurred. Putting together Eqs.~\eqref{mean_weight} and \eqref{weight_decomp}, the bound decomposes as
\begin{equation}
    \bar p_L=\bar W\cdot\dfrac{\Pr(P\neq\mathbb{I})}{p_L}\cdot p_L,
\end{equation}
which isolates the two sources of looseness in Theorem \ref{main_theorem} as exact multiplicative factors.

The factor $\Pr(P\neq\mathbb{I})/p_L$ quantifies the first inequality and reflects the intrinsic immunities of the target circuit. These include logical errors that the traps register but which leave the computation unaffected, such as logical $Z$ errors carried by an injection ancilla at its measurement slice, which commute with its readout, or, for circuits with classical output, data errors commuting with the measured logical observables. Its size is therefore set by the readout structure of the target circuit. It is smallest when the output is a logical quantum state, for which every residual error corrupts the stored quantum information, and larger when the output consists of classical measurement outcomes. In either case, this factor is a circuit-dependent constant, and it is not accessible without knowledge of $p_L$ itself, which is precisely the quantity the protocol is designed to avoid having to measure.

The factor $\bar W$ quantifies the second inequality and reflects how errors propagate. If logical errors were confined to individual logical qubits, then $\bar W=1$, and the union bound would be saturated. Entangling logical operations, however, can spread errors across blocks, so that multi-qubit logical errors are generated at the same order in the physical error rate as single-qubit ones. Thus, $\bar W>1$ in any circuit with a non-trivial entangling structure, growing with depth and connectivity. The overcounting is nonetheless bounded, since $W(F)\leq K$ for every fault configuration and hence $\bar W\leq K$, so that the union bound can inflate the estimate by at most the logical footprint of the circuit. In the fault-tolerant regime, $\bar W$ approaches a circuit-dependent constant, and is expected to grow slowly only when the physical error rate approaches the code fault-tolerance threshold.

Overall, $\bar p_L$ remains proportional to the true logical error rate up to these two bounded factors, tracking the target's true logical error rate rather than diverging from it.

\section{Certification protocols}\label{protocols}

Having established the upper bound of Theorem \ref{main_theorem}, we now present
the two operational protocols that evaluate it in practice, which differ in the question they answer. Direct estimation (Sec.~\ref{direct_estimation}) reconstructs the bound itself by estimating every trap failure rate, which is useful when the value of the bound is wanted, but its cost grows as the device improves. Threshold certification (Sec.~\ref{cert_protococol}) instead decides whether the bound lies below a predefined tolerance, at a cost independent of the error rates of the device, providing the operational form of the magic-free certification protocol.

The whole formalism is flexible regarding the circuit it certifies. Applying the construction to a single injection gadget benchmarks that gadget, while applying it to a full logical algorithm certifies the entire computation, with no change in the protocol. The experimental cost of the protocols is set by two resources: the number of distinct trap configurations to be compiled, and the total number of shots. Complete coverage of the logical Pauli basis requires $3K$ trap configurations, with the correction bitstring sampled per shot.

\subsection{Direct estimation of the bound}\label{direct_estimation}

Since each trap execution yields a binary
outcome, every shot constitutes an independent Bernoulli trial, with the independence across shots following from A1, and the estimation admits a transparent statistical treatment. After $M$ shots, the estimated failure rate carries a standard deviation $\sigma = \sqrt{p(\mathcal{T}_i)[1-p(\mathcal{T}_i)]/M}$, which in the fault-tolerant regime $p(\mathcal{T}_i) \ll 1$ reduces to $\sigma \approx \sqrt{p(\mathcal{T}_i)/M}$. Estimating $p(\mathcal{T}_i)$ to a
relative precision $\epsilon = \sigma / p(\mathcal{T}_i)$ therefore requires
\begin{equation}\label{direct_shots}
    M \approx \frac{1}{\epsilon^2 p(\mathcal{T}_i)}
\end{equation}
shots per trap circuit. The upper bound $\bar p_L$ is then constructed from the $3K$ estimated rates through Eq.~\eqref{bound}, with a relative
uncertainty no larger than $\epsilon$.

Equation \eqref{direct_shots} exposes the sample-complexity bottleneck inherent to benchmarking highly reliable hardware. For a code of distance $d$ operating at physical error rate $p$, the trap failure rates are suppressed as $p(\mathcal{T}_i) =C_i p^{\lfloor (d+1)/2 \rfloor}$ in the fault-tolerant regime, with circuit-dependent prefactors $C_i$, so that the total number of shots required to resolve the bound at fixed relative precision
scales as
\begin{equation}\label{direct_scaling}
    M_{\text{tot}} \sim \mathcal{O}\!\left( \frac{K}{p^{\lfloor (d+1)/2\rfloor}}\right).
\end{equation}
Although linear in the logical footprint $K$, this cost grows exponentially
with the code distance as hardware improves. Consequently, the better the quantum device is, the more
shots are needed to resolve how good it is. Direct estimation is therefore best suited to near threshold hardware, where the
trap rates remain large enough to be resolved in a resource-efficient manner. Certifying advanced fault-tolerant hardware, whose logical error rates are by design
vanishingly small, calls for a different approach. Rather than resolving the exact value of a quantity that is exponentially close to zero, it suffices to certify that it lies below a tolerance predefined by the user, at a cost that, as shown next, depends on neither the code nor the hardware error rate. This is the subject of the following subsection.

\subsection{Threshold certification protocol}\label{cert_protococol}

As shown above, directly reconstructing the exact logical error rate $p(\mathcal{T}_i)$ of a highly reliable trap circuit requires a prohibitive number of experimental shots, scaling as $\mathcal{O}(1/p(\mathcal{T}_i))$. To circumvent this sample complexity bottleneck for practical hardware evaluations, we propose a certification protocol based on a statistical hypothesis test, which constitutes the operational form of magic-free certification. Rather than reconstructing the exact upper bound, we aim to certify that the logical error rate is strictly below a predefined error tolerance $p^*_L$ at a confidence level $1-\delta$, where $\delta$ is the significance level, i.e., the false positive rate at $\bar p_L=p_L^*$.

The key observation is that the upper bound of Theorem \ref{main_theorem} depends
on the individual trap circuits only through their sum. It is therefore not
necessary to resolve each failure rate separately, but it suffices to constrain the
average failure rate of the complete trap set,
\begin{equation}\label{qbar}
    \bar f = \frac{1}{3K} \sum_{i=1}^{3K} p(\mathcal{T}_i) = \frac{2 \bar p_L}{3K}.
\end{equation}
Since $\bar f$ is proportional to $\bar p_L$, certifying $\bar p_L < p_L^*$ is equivalent to certifying that the average trap failure rate lies below the corresponding local security threshold
\begin{equation}\label{qstar}
    f^* = \frac{2 p_L^*}{3K}.
\end{equation}
This reduces the certification of the whole logical circuit to a single
hypothesis test on a single parameter.

Assuming no prior knowledge about the hardware, all logical error channels are taken to be equally damaging, and the trap set is sampled uniformly. The total
budget of $M_C$ certification shots is distributed evenly among the $3K$ trap configurations interleaved at random with the standard executions of the target
circuit. Let $m$ denote the total number of failures registered across the whole trap set. Its expected value,
\begin{equation}\label{m_expected}
    \lambda = M_C\, \bar f = \frac{2 M_C}{3K}\, \bar p_L,
\end{equation}
is directly proportional to the bound we wish to certify, irrespective of how the logical errors are distributed among the individual trap circuits.

We define the null hypothesis as $\bar f \ge f^*$, i.e., the device fails to meet the target tolerance. To account for the probabilistic nature of quantum noise, we allow up to $w$ failures. If $m>w$ failures are detected, the device is declared uncertified, whereas if $m\leq w$, the null hypothesis is rejected and the target circuit is certified. Because trap executions are independent Bernoulli trials, the aggregated failure count follows a Poisson binomial distribution, whose rates $p(\mathcal{T}_i)$ need not be equal. By Le Cam's theorem \cite{LeCam1960}, its total variation distance from a Poisson distribution of the same mean $\lambda=M_C\bar f$ is at most $\lambda\cdot\max_ip(\mathcal{T}_i)$, which is negligible in the fault-tolerant regime. The certification probability therefore reads
\begin{equation}\label{P_C}
    P_C = \Pr(m\leq w)= F_P\left(w;\, M_C \bar f\right),
\end{equation}
where $F_P(w; \lambda) = \sum_{m=0}^{w} \lambda^m e^{-\lambda}/m!$ is the
Poisson cumulative distribution function (CDF). Since $F_P$ decreases monotonically with its mean, $P_C$ decreases monotonically with $\bar p_L$. Thus, the worse the device, the less likely it is to be certified.

The protocol is then fully specified by two conditions. The first one enforces its security: at the tolerance $\bar p_L = p_L^*$, the probability of
erroneously certifying a device that does not meet the tolerance must not
exceed the significance level $\delta$,
\begin{equation}\label{security}
    F_P\!\left(w;\frac{2 M_C p_L^*}{3K}\right) \le \delta.
\end{equation}
The second condition enforces its yield: any device operating safely below the predefined tolerance $p_L^*$, at an operational target $\bar p_L\leq rp_L^*$, with operational ratio $0<r<1$, must be certified with probability at least $1-\gamma$,
\begin{equation}\label{yield}
    F_P\left(w;\frac{2 M_C r p_L^*}{3K}\right) \geq 1 - \gamma,
\end{equation}
where $\gamma$ is the false negative rate at the operational target $rp_L^*$. Between the yield boundary $r\,p_L^*$ and the security threshold $p_L^*$ lies
the \emph{indifference zone} of the protocol, in which no statistical guarantee is enforced and the outcome may go either way depending on statistical fluctuations. Its width is fixed by the operational ratio alone (it is the interval $(r\,p_L^*,\, p_L^*)$ and it can be narrowed at will by increasing $r$, at a shot cost quantified below. Moreover, since $m$ is sensitive only to the average trap
failure rate, the guarantees are distribution-free. Hardware in which faults are evenly spread across all traps and hardware in which all errors concentrate on a single trap are certified with identical statistical guarantees.

The number of allowed failures $w$ and the total number of certification shots $M_C$ are obtained by finding the smallest pair $(w,M_C)$ that simultaneously satisfies Eqs.~\eqref{security} and \eqref{yield}. As detailed in Appendix \ref{test_calibration}, this calibration is scale-invariant. The failure tolerance $w$ depends only on the statistical parameters $\delta$ and $\gamma$ and the operational ratio $r$. It is obtained from the expected failure counts $\lambda^* = \frac{1}{2} F_{\chi^2}^{-1}(1-\delta; 2(w+1))$ and $\lambda^t = \frac{1}{2} F_{\chi^2}^{-1}(\gamma;2(w+1))$ \cite{Johnson_Kotz_Kemp_2005}, with $F_{\chi^2}^{-1}(\cdot;\cdot)$ the inverse chi-squared CDF, through the condition $\lambda^t / \lambda^* \geq r$. Only once $w$ is fixed does the physical scale enter, in order to compute the required number of shots through $M_C = \lceil \lambda^* / f^* \rceil$.

The interplay of the security and yield conditions is illustrated in Fig.~\ref{main_figure}(d) for $r=0.5$. The curve crosses the significance level $\delta$ at the security threshold $\bar p_L=p_L^*$, enforcing Eq.~\eqref{security}, and meets the yield level $1-\gamma$ at the operational target $\bar p_L=rp_L^*$, enforcing Eq.~\eqref{yield}. We stress that, once plotted against $\bar p_L/p_L^*$, the curve is universal, since its shape depends only on the statistical parameters $\delta$, $\gamma$ and $r$, while the circuit being certified enters solely through the shot budget $M_C$ and through the value of $\bar p_L$ that the measurements assign to the device.
 
\begin{figure*}[t!]
    \centering
    \includegraphics[width=\textwidth]{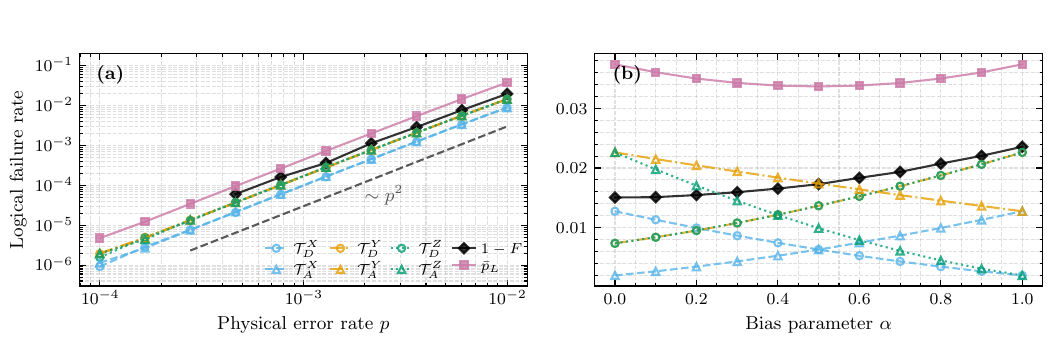}
    \caption{Upper bound for the logical circuit $HT\ket{+}$ encoded in the Steane code. The subscripts $D$ and $A$ denote traps associated with the logical data qubit ($\mathcal{T}_D^O$) and the logical ancilla qubit ($\mathcal{T}_A^O$), respectively, and the superscript $O\in\{X,Y,Z\}$ denotes the measured logical Pauli operator. Each data point is estimated from $M=10^8$ shots per circuit configuration. (a) Scaling of logical failure rates as a function of the physical depolarizing error rate $p$. The derived upper bound $\bar p_L$ is compared against the output-state infidelity reconstructed via full QST, which lower-bounds the
    total logical Pauli error probability. Both quantities exhibit the expected
    $\sim p^2$ fault-tolerant scaling for a distance-3 code, with the bound
    tracking the infidelity up to an approximately constant factor. (b) Trap sensitivities under biased Pauli noise, ranging from pure $X$ ($\alpha=0$) to pure $Z$ ($\alpha=1$), at a fixed total physical error rate of $p=0.01$. The data and ancilla traps respond to opposite ends of the sweep, reflecting the direction in which each fault type propagates through the gadget, while the bound remains above the reconstructed infidelity across the entire bias range.}
    \label{msi_numerics}
\end{figure*}

Finally, we evaluate the certification shot budget. Writing
$\lambda^* = M_C f^*$ for the expected number of failures at the security
threshold, the shot count required for certification is
$M_C = (3K / 2 p_L^*)\, \lambda^*$, which, expanding $\lambda^*$ in the
test's parameters as detailed in Appendix \ref{shot_scaling}, yields
\begin{equation}\label{MC_scaling}
    M_C \sim \mathcal{O}\left( \frac{K}{p_L^*}
    \frac{\left(\sqrt{-\ln\delta} + \sqrt{-\ln\gamma}\right)^2}{(1-r)^2}
    \right).
\end{equation}
The protocol incurs a linear penalty in the circuit size $K$, since the local security threshold $f^*$ is diluted among the $3K$ trap configurations, and a
quadratic penalty for sharpening the operational ratio towards $r \to 1$ (Fig.~\ref{main_figure}(e)). In exchange, the sample complexity is completely decoupled from the vanishing hardware error rates. By selecting a practical tolerance $p_L^*$ well above the bound $\bar p_L$ the device achieves, the certification cost is set by the question being asked rather than by the quality of the device answering it, rendering rigorous circuit certification viable precisely in the regime where direct estimation becomes prohibitive. All of this is achieved without executing a single non-Clifford operation, so that certification consumes no magic states.

\section{Numerical simulations}\label{numerics}

In this section, we illustrate the theoretical framework and operational protocols of Secs.~\ref{framework} and \ref{protocols} with numerical simulations under circuit-level noise. These simulations are illustrative rather than confirmatory, since the bound of Theorem~\ref{main_theorem} is already proven, and the statistical guarantees of the certification protocol, namely its security, yield and shot budget, follow analytically from Sec.~\ref{protocols} and depend only on $\delta$, $\gamma$, $r$, $K$ and
$p_L^*$, so that simulating larger circuits would only reproduce known results without adding information. The simulations show instead how the bound behaves in practice, how the individual traps respond to different noise, and
how the test is calibrated for a multi-qubit algorithm. To this end, we consider two fault-tolerant layouts with logical qubits encoded in the Steane
code \cite{Steane1996}: an MSI gadget, small enough for the error of the target to be reconstructed independently by state tomography, and the QFT. The noise channels applied in the simulations are the stochastic Pauli channels assumed
in A1. Trap circuits are then efficiently simulated using the stabilizer simulator Stim \cite{gidney2021stim}, whereas to compute the infidelity with respect to the target state in the MSI gadget we perform QST simulations using Tsim \cite{tsim2026}.

\subsection{Bounding a magic state injection gadget}\label{bounding_msi}

To characterize the performance of a non-Clifford gadget, we consider the logical circuit $HT\ket{+}$, where we apply a $T$ gate via MSI to the Pauli-$X$ eigenstate $\ket{+}$, followed by a Hadamard gate $H$ in order to map the phase information into population. To ensure a complete coverage of logical Pauli errors, we implement the three distinct trap configurations for both the logical data block (labeled $\mathcal{T}_D^O$) and the logical ancilla block (labeled $\mathcal{T}_A^O$). We then evaluate the upper bound under two distinct noise environments: depolarizing noise to show the fault-tolerant scaling of the upper bound (Fig.~\ref{msi_numerics}(a)), and biased noise to test the resilience of the upper bound under fluctuations of the noise model (Fig.~\ref{msi_numerics}(b)).

To implement QEC both in the trap and QST shots, we infer the parities of the relevant stabilizers from the measured observables. The syndromes are decoded with a lookup table (LUT), and the measurement outcomes are left as measured or flipped accordingly. All circuits are decoded with the same layer-wise LUT scheme, in which corrections are recorded in the Pauli frame and not propagated through the logical gates. Following A4, all reported failure rates are relative to this decoder. Following Sec.~\ref{traps_subsection}, each trap is executed in both of its correction-bitstring variants, with $M/2$ shots each and failures combined uniformly. For trap circuits, logical failure is declared if, after QEC, the logical observable outcome does not match the expected deterministic value. For QST simulations we compute the infidelity of the output state, $1-F=\frac{1}{2}\left(1-\mathbf{r}_{\mathrm{id}}\cdot\mathbf{r}_{\mathrm{meas}}\right)$,
where $\mathbf{r}_{\mathrm{id}}=(0,-1,1)/\sqrt{2}$ is the ideal Bloch vector, and $\mathbf{r}_{\mathrm{meas}}=(\langle X\rangle, \langle Y\rangle,\langle Z\rangle)$ is the reconstructed Bloch vector. For a stochastic Pauli channel acting on a pure target state, the infidelity
lower-bounds the probability that the data output carries a non-trivial logical Pauli error, so that quantities compared in Fig.~\ref{msi_numerics} are
ordered as $1-F\leq\Pr(P\neq\mathbb{I})\leq\bar p_L$. The vertical gap between $\bar p_L$ and $1-F$ therefore combines three distinct contributions: the difference between infidelity and Pauli error probability, the intrinsic immunities of the target circuit, and the union-bound overcounting, the latter two discussed in Sec.~\ref{bound_tightness}.

For the baseline fault-tolerant scaling analysis (Fig.~\ref{msi_numerics}(a)), we apply a circuit-level depolarizing noise model. Fault-tolerant state preparation is modeled by an ideal initialization followed by a layer of single-qubit depolarizing noise with probability $p_p$, independent of the prepared basis, consistent with A3. Two-qubit depolarizing noise is applied after every CNOT gate at a rate $p_g$, alongside measurement errors with probability $p_m$. Single-qubit gates are taken to be noiseless. We assume a uniform physical error rate $p=p_p=p_g=p_m$. To accurately construct the upper bound (UB) and minimize the impact of quantum projection noise in QST simulations, the number of shots for each circuit configuration is set to $M=10^8$. As illustrated in Fig.~\ref{msi_numerics}(a), both the derived upper bound $\bar p_L$ and the logical infidelity obtained via QST exhibit the $\mathcal{O}(p^2)$ fault-tolerant scaling expected for a distance-3 code, with the upper bound remaining strictly above the reconstructed infidelity across all physical error rates $p$. The QST curve terminates however at $p\approx5\cdot10^{-4}$. Below this regime, resolving the hardware-induced infidelity against quantum projection noise requires a number of tomographic shots that grows as the logical error rate vanishes, which is precisely the bottleneck that deterministic traps remove. By replacing these tomographic measurements with deterministic trap circuits, the certification upper bound successfully bypasses this statistical sampling bottleneck and tracks the vanishing error rate accounting only for hardware faults.

To demonstrate the robustness of the upper bound under different noise profiles, Fig.~\ref{msi_numerics}(b) evaluates the same injection gadget at a fixed total physical error rate $p=p_X+p_Z=0.01$ while varying the noise bias fraction $\alpha=p_Z/p$. Here, the standard depolarizing channels from the baseline simulations are replaced by a structured biased noise model. For state preparation, errors are applied via a single-qubit Pauli channel restricted to $X$ errors with probability $p_X=(1-\alpha)p$ and $Z$ errors with probability $p_Z=\alpha p$. Two-qubit noise is applied after each CNOT gate by distributing pure $X$-type errors equally among $IX, XI,$ and $XX$ errors, which occur with a total probability $p_X$, and pure $Z$ errors equally among the $IZ, ZI,$ and $ZZ$ channels, taking place at a total rate $p_Z$. All cross-terms and $Y$ errors are set to zero, and the single-qubit gates remain ideal. Finally, measurement errors are applied at a constant rate $p$, completely independent of the bias parameter $\alpha$.

Figure \ref{msi_numerics}(b) shows that the data and ancilla traps respond to opposite ends of the sweep, reflecting the direction in which each fault type propagates through the gadget. This behavior is best understood by examining the limiting cases of pure $X$ ($\alpha=0$) and pure $Z$ ($\alpha=1$) noise. Under pure $X$ noise, physical $X$ errors on the data block are transformed into logical $Z$ errors by the Hadamard gate just before the measurement. Consequently, only the data qubit's $X$ and $Y$ traps, which are sensitive to logical $Z$ errors, detect these faults, whereas the data $Z$-trap registers only bias-independent measurement errors. Furthermore, the transversal CNOT gates propagate physical $X$ errors forward from the control (data) to the target (ancilla), yielding a logical $X$ error at the ancilla's measurement slice. This fault accumulation causes the ancilla's $Y$ and $Z$ traps to have higher failure rates than their data counterparts at $\alpha=0$. Conversely, under pure $Z$ noise ($\alpha=1$), physical $Z$ errors propagate backward from the target to the control. The data block's Hadamard gate then transforms these accumulated $Z$ faults into logical $X$ errors. As a result, the data qubit's $Y$ and $Z$ traps become susceptible to the noise, surpassing the failure rates of the corresponding ancilla traps. 

While tracking these specific propagation channels provides physical insight, the primary purpose of the upper bound is to bypass the need for such exhaustive error modeling. This is particularly evident when analyzing the reconstructed infidelity, which exhibits a pronounced asymmetry across the bias sweep. This asymmetric behavior is driven by the feedforward mechanics of the injection gadget: under pure $X$ noise ($\alpha=0$), the ancilla block accumulates logical $X$ errors that flip the logical measurement outcome, triggering an incorrect feedforward operation that degrades the final state fidelity. Conversely, under pure $Z$ noise ($\alpha=1$), the ancilla exclusively suffers logical $Z$ errors, which commute with the $Z$-basis measurement and leave the feedforward operation completely unaffected. Consequently, the logical failure rate at $\alpha=1$ is driven almost entirely by the data qubit's logical $X$ errors, explaining why the reconstructed infidelity directly converges with the data trap failure rates at the rightmost edge of the plot. The upper bound nonetheless remains above the QST curve across the entire bias range, as Theorem~\ref{main_theorem} guarantees for any stochastic Pauli noise, without requiring any characterization of the noise structure.

\subsection{Certifying a fault-tolerant QFT}

\begin{figure*}[t!]
    \centering
    \includegraphics[width=\textwidth]{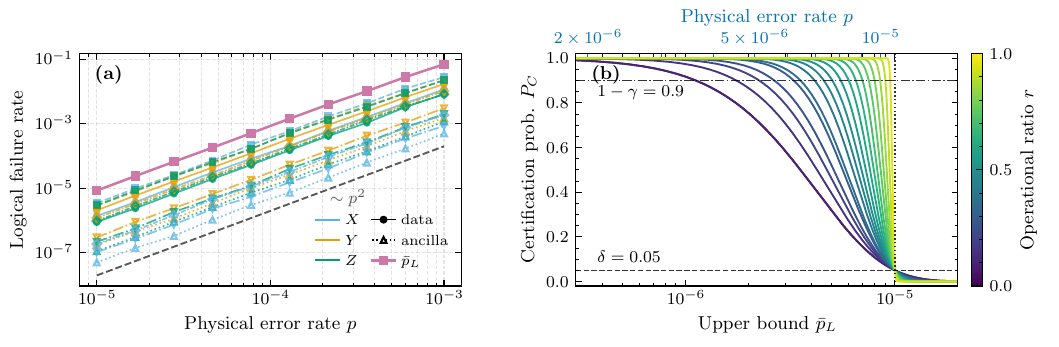}
    \caption{Certification of a fault-tolerant QFT on $k_D=2$ logical data qubits and $k_T=3$ injections encoded in the Steane code. (a) Failure rates of the $15$ trap configurations as a function of the physical depolarizing error rate $p$, together with the upper bound
$\bar p_L$ assembled from them through Theorem \ref{main_theorem}. Shots are assigned to each trap based on a preliminary estimate of its quadratic prefactor, so that each rate is resolved to a fixed relative precision $\epsilon=0.05$. Colors
denote the measured logical Pauli operator and marker fill
distinguishes data blocks (filled) from injection ancillas (empty), whereas different linestyles signal different qubits. The dashed grey
line is a $\sim p^2$ guide. Each trap is executed in its $2^{k_T}$
correction-bitstring variants with failures combined uniformly. (b)
Certification probability $P_C$ as a function of the upper bound $\bar p_L$,
for operational ratios $r$ ranging from $0.05$ to $0.95$ (color bar),
calibrated at $p_L^*=10^{-5}$ (dotted line), $\delta=0.05$ (dashed line) and
$1-\gamma=0.9$ (dash-dotted line). Each curve crosses the significance level at $\bar p_L=p_L^*$ and the yield level holds at $\bar p_L=r\,p_L^*$, delimiting the
indifference zone. The upper axis gives the physical error rate of the
simulated device inferred from the quadratic scaling fitted in (a).}
    \label{qft_fig}
\end{figure*}

To demonstrate the certification protocol at the scale of a full logical
algorithm, we consider a QFT acting on $k_D=2$
logical data qubits. Its fault-tolerant implementation requires $k_T=3$ magic
state injections (two $T$ gates and one $T^\dagger$ gate), so that the total
logical footprint is $K=5$. Complete characterization of the
upper bound requires $15$ distinct trap configurations, each executed in
its $2^{k_T}=8$ correction-bitstring variants with failures combined uniformly.
All logical qubits are encoded in the Steane code, and we apply the same circuit-level depolarizing noise model as in the previous subsection, with a
uniform physical error rate $p$. To estimate the trap failure rates, and unlike the fixed number of shots of Sec.~\ref{bounding_msi}, shots are now allocated to each trap from a preliminary estimate, so that every rate is resolved to a fixed relative precision $\epsilon=0.05$.

Because the QFT circuit contains a significantly larger gate count than a
single injection gadget, faults have many more opportunities to propagate, and
a single error correction step at the end of the circuit is insufficient. We
therefore interleave the logical operations with layers of Steane error
correction \cite{Steane_1997,Postler2024}, which extract the syndromes fault-tolerantly, on every logical block within its live window. Data blocks undergo syndrome extraction at every layer until their final
measurement, whereas each injection ancilla undergoes it only up to its
measurement slice, mirroring the target circuit. During each error correction
cycle, the blocks interact transversally with logical ancillas prepared in the
$\ket{0}$ and $\ket{+}$ states, which are subsequently measured in the $X$ and $Z$ bases to extract
the $Z$- and $X$-error syndromes respectively. The preparation of these syndrome ancillas is
subject to the same noise model as the data blocks. We stress that the error
correction schedule is part of the circuit location set shared by target and
traps. Every block is corrected in every trap configuration, irrespective of
which qubit that trap targets, since faults occurring during syndrome
extraction on the remaining blocks propagate onto the targeted block in
the target circuit and must therefore be sampled by its traps. All circuits
are decoded with the same layer-wise LUT scheme, in which corrections are
recorded in the Pauli frame and are not propagated through the logical gates, as in the injection gadget simulation and thus following A4.

Figure \ref{qft_fig}(a) shows the failure rates of the $15$ trap
configurations, together with the upper bound $\bar p_L$ assembled from them
through Theorem \ref{main_theorem}. All trap rates exhibit the $\sim p^2$ scaling expected for a distance-3 code across the range of physical error rate
considered, confirming that the full circuit, including its error correction schedule and mid-circuit ancilla measurements, remains fault-tolerant at algorithmic scale. The spread among the individual rates reflects the position of each block in the circuit and the direction in which
faults propagate towards it, with the injection ancillas, whose live windows are shorter, generally failing at lower rates than the data blocks.

No independent reconstruction of the output state is performed at this scale. Beyond the shot cost, which grows as the logical error rate vanishes, such a reconstruction would only reproduce on a larger circuit the ordering already established in Sec.~\ref{bounding_msi}, without probing any additional feature of the protocol. The trap failure rates and the bound assembled from them are therefore the object of study, and the remainder of this subsection illustrates how the certification protocol operates on them.

Once the trap set is fixed, the protocol is completely specified by the four
statistical parameters $p_L^*$, $\delta$, $\gamma$ and $r$ and the trap count. Targeting $p_L^*=10^{-5}$ with $\delta=0.05$ and
$\gamma=0.1$, the calibration procedure of Appendix \ref{test_calibration} yields $(w,M_C)=(4,6.9\cdot10^{6})$ at $r=0.25$, $(17,1.9\cdot10^{7})$ at $r=0.5$, and $(102,9.0\cdot10^{7})$ at $r=0.75$, with the shots distributed uniformly among the trap configurations and their correction-bitstring variants. These numbers illustrate the resolution penalty of Eq.~\eqref{MC_scaling}, since sharpening the operational ratio towards $r\to1$ inflates both the failure tolerance and the shot budget as $(1-r)^{-2}$, while the dependence on the security threshold $1/p_L^*$ remains decoupled from the vanishing hardware error rates.

Figure \ref{qft_fig}(b) shows the resulting certification probability $P_C$ as a function of the upper bound $\bar p_L$, for operational ratios ranging from $r=0.05$ to $r=0.95$, complementing the single curve of Fig.~\ref{main_figure}(d) by showing how the transition sharpens as $r\to1$. Every curve crosses the significance level $\delta$ at $\bar p_L=p_L^*$ and meets the yield level $1-\gamma$ at $\bar p_L=rp_L^*$, so that the indifference zone between them narrows as $r$ grows. The upper axis translates the bound into the physical error rate of the simulated device through the quadratic scaling
fitted in Fig.~\ref{qft_fig}(a). We note that this conversion to physical error rates is a property of the distance-3 Steane code under this particular (unoptimized) LUT decoder and at the depth required by the QFT rather than of the certification protocol, whose cost $\mathcal{O}(K/p_L^*)$ is unchanged by moving to a code of higher distance.

\section{Discussion and outlook}\label{discussion}

We have introduced magic-free certification, a protocol that decides whether a fault-tolerant device executes a universal logical circuit below a tolerance fixed in advance, using only stabilizer circuits run on the same device as the target. Its guarantees follow analytically from the bound of Theorem~\ref{main_theorem} and from the calibration of the test, so that the simulations of Sec.~\ref{numerics} can illustrate them but not extend them. What remains open is how the protocol behaves on real hardware, and an experimental implementation is therefore the natural next step. The protocol requires no operation beyond those already used by the target circuit, except for the substitution of its magic states, together with random interleaving of trap and target executions and, when corrections are applied physically, uniform sampling of the correction bitstrings.

The logical circuits needed for such a demonstration are already being run.
Logical non-Clifford gates have been implemented at break-even with trapped ions \cite{Dasu2025}, logical magic states have been prepared and distilled with superconducting circuits and neutral atoms \cite{Gupta2024,Ye2023,SalesRodriguez2025}, and neutral-atom processors have run logical algorithms on tens of logical qubits \cite{Bluvstein2023}. A three-qubit QFT encoded in the Steane code, of which the circuit of
Sec.~\ref{numerics} is a smaller instance, has moreover been executed on a trapped-ion processor \cite{mayer2024benchmarkinglogicalthreequbitquantum}. There, performance was assessed through component benchmarks and fidelity lower bounds, and the
component errors were found to account for only part of the error of the full circuit. Magic-free certification would instead return a verdict on the full circuit at a chosen tolerance.

Beyond such demonstrations, the protocol can be used in several ways. Since the traps never require the target circuit to be run, a device can be certified before it is put into operation, for instance prior to delivery, by executing the traps across the range of operating conditions, such as temperature or humidity, under which the target is later meant to run, rather than interleaved with it as assumed throughout this work. Once the device is in operation, repeating the test over time monitors whether it keeps meeting its tolerance as its calibration drifts. For circuits built from repeated primitives, such as the arithmetic of the factoring algorithm discussed in Sec.~\ref{main_results}, each primitive can be certified separately, locating any component that fails the tolerance. Moreover, since the verdict refers to a tolerance chosen by the user rather than to any feature of the hardware, running the same test on different devices compares them on equal terms, in the spirit of cross-platform verification \cite{Elben2020}.

Several extensions of the framework are worth pursuing. First, trap executions discard most of the information they produce, since each returns a single bit, and the syndrome records collected along the way are used only for decoding. The per-qubit and per-Pauli rates could identify which logical qubit or error channel is responsible for a failed certification, at the cost of additional shots to resolve them individually, and the three rates of each qubit determine the weights of its effective single-qubit logical Pauli channel. Traps stabilized by multi-qubit operators, introduced in the adversarial setting in Ref.~\cite{Sater2026}, would further give access to correlations between logical qubits, and combining trap outcomes with the syndrome data exploited by noise-learning methods \cite{Wagner2022paulichannelscanbe,Wagner2023,Zheng2026,xiao2026,Pancotti2026} could deliver, from the same executions, both a certificate and a characterization of the logical noise. Second, A2 delegates the certification of the magic-state preparation stage to an independent procedure. Integrating a characterization of the injected states into the protocol, for instance through dedicated benchmarks of logical magic states run on the same device \cite{Lee2025magicbenchmark}, would extend the certificate to the resource states themselves, and would allow the framework to be extended to architectures that do not implement these resources fault-tolerantly, such as those that error-correct Clifford operations but inject non-Clifford rotations without distillation \cite{Akahoshi2024,Toshio2025,Ismail2026}. In each case, the traps remain the same object, stabilizer circuits sharing the target's noise, and what changes is only how much of the information they carry is used.

\begin{acknowledgements}

GPS acknowledges T. Hashizume for useful discussions. GPS and DJ acknowledge support from DFG project ‘Quantencomputing mit neutralen Atomen’ (JA 1793/1-1, Japan-JSTDFG-ASPIRE 2024). SP and DJ acknowledge support from the Hamburg Quantum Computing Initiative (HQIC) project EFRE. The EFRE project is co-financed by ERDF of the European Union and by “Fonds of the Hamburg Ministry of Science, Research, Equalities and Districts (BWFGB)”. DJ acknowledges support from the Cluster of Excellence ‘Advanced Imaging of Matter’ of the Deutsche Forschungsgemeinschaft (DFG) - EXC 2056 - project ID 390715994, the European Union’s Horizon Programme (HORIZONCL42021- DIGITALEMERGING-02-10) Grant Agreement 101080085 QCFD, and funding from the Federal Ministry of Research, Technology and Space (BMFTR) under the grant BeRyQC.

\end{acknowledgements}

\appendix

\section{Upper bound under residual coherent errors}\label{coherent_appendix}

Here, we relax A1, allowing the channel at each location to be an arbitrary completely positive trace-preserving (CPTP) quantum map rather than a stochastic Pauli channel, while retaining A2--A4. With these considerations, we generalize the upper bound to the logical error rate of the target circuit $p_L$ from Eq. \eqref{bound}, and show that it degrades additively rather than collapsing entirely.

We consider the target circuit with logical footprint $K=k_D+k_T$ and its trap variants to be composed of a set of $L$ ideal locations (gates, idle periods, and measurements), with each location $\ell$ followed by an arbitrary CPTP channel $\mathcal{E}_\ell$ describing noise. We denote the Pauli-twirled version of the channel $\mathcal{E}_\ell$ as $\mathcal{P}_\ell$.

To compare the actual error channel at location $\ell$ with its twirled counterpart, we make use of the diamond norm. The diamond norm between two quantum channels $\mathcal{F}$ and $\mathcal{G}$, each acting on states of a Hilbert space $\mathcal{H}_S$, is defined as
\begin{equation}\label{diamond_norm}
    \|\mathcal{F}-\mathcal{G}\|_\diamond \;=\; \max_{\sigma} \big\|(\mathcal{F}\otimes \mathbb{I}_R)(\sigma) - (\mathcal{G}\otimes \mathbb{I}_R)(\sigma)\big\|_1 ,
\end{equation}
where the identity $\mathbb{I}_R$ acts on states of a reference Hilbert space $\mathcal{H}_R$ with dimension $\dim(\mathcal{H}_R)=\dim(\mathcal{H}_S)$, $\sigma$ is a density matrix on the joint Hilbert space $\mathcal{H}_S\otimes\mathcal{H}_R$, and $\|X\|_1=\mathrm{Tr}\sqrt{X^\dagger X}$ is the trace norm. We then define the coherent residual at location $\ell$ to be
\begin{equation}
    \eta_\ell=\dfrac{1}{2}\|\mathcal{E}_\ell-\mathcal{P}_\ell\|_\diamond,
\end{equation}
and the total coherent residual as
\begin{equation}
    \eta=\sum_{\ell=1}^L\eta_\ell.
\end{equation}

Let $\mathcal{C}$ denote the quantum channel describing the successive action of every ideal circuit location followed by its associated noise channel $\mathcal{E}_\ell$, acting on an arbitrary input state $\rho$. We define $\mathcal{C}_j$ to be the circuit in which channels associated with locations $1,\dots,j$ have been replaced by their twirls, while those of locations $j+1,\dots,L$ remain unaltered, so that $\mathcal{C}_0=\mathcal{C}$ is the actual noisy circuit and $\mathcal{C}_L$ is its fully twirled counterpart, to which Theorem~\ref{main_theorem} applies. The action of the channel $\mathcal{C}_j$ on an arbitrary state $\rho$ is thus given by
\begin{equation}
    \mathcal{C}_j(\rho)=\mathcal{D}_j\left[\mathcal{P}_j(\rho_j)\right],
\end{equation}
where $\rho_j$ denotes the state produced by the first $j$ locations of the circuit, and $\mathcal{D}_j$ is the part of $\mathcal{C}_j$ that contains the remaining locations after the twirled channel $\mathcal{P}_j$ has been applied.

To quantify the degradation of the upper bound under residual coherent noise, we must bound the quantity
\begin{equation}\label{coherent_diff}
    \dfrac{1}{2}\big\|\mathcal{C}(\rho) - \mathcal{C}_L(\rho)\big\|_1.
\end{equation}
Consecutive members of the family of $\mathcal{C}_j$ channels differ only at a single error channel, which is what allows the deviation to be bounded one location at a time:
\begin{align}\label{local_bound}
    \dfrac{1}{2}\big\|\mathcal{C}_{j-1}(\rho) - \mathcal{C}_j(\rho)\big\|_1 &= \dfrac{1}{2}\big\|\mathcal{D}_j\left[\mathcal{E}_{j}(\rho_j) - \mathcal{P}_j(\rho_j)\right]\big\|_1 \notag\\
    & \leq\dfrac{1}{2}\big\|\mathcal{E}_{j}(\rho_j) - \mathcal{P}_j(\rho_j)\big\|_1 \notag \\
    & \leq \dfrac{1}{2}\big\|\mathcal{E}_{j} - \mathcal{P}_j\big\|_\diamond = \eta_j,
\end{align}
where in the first inequality we have used the contractivity of CPTP maps, and in the second inequality we have used the diamond norm definition from Eq. \eqref{diamond_norm}. With this, we can upper-bound Eq. \eqref{coherent_diff}:
\begin{align}\label{coherent_bound}
    \dfrac{1}{2}\big\|\mathcal{C}(\rho) - \mathcal{C}_L(\rho)\big\|_1&=\dfrac{1}{2}\left\|\sum_{j=1}^L\left[\mathcal{C}_{j-1}(\rho) - \mathcal{C}_j(\rho)\right]\right\|_1 \notag \\
    &\leq \dfrac{1}{2}\sum_{j=1}^L\left\|\mathcal{C}_{j-1}(\rho) - \mathcal{C}_j(\rho)\right\|_1 \notag \\
    &\leq\sum_{j=1}^L\eta_j=\eta.
\end{align}

Labeling the logical error rate of the target circuit that arises from all noise channels being perfectly Pauli-twirled as $p_L^P$, and reserving $p_L$ for the actual logical error rate, from Eq. \eqref{coherent_bound} we infer that 
\begin{equation}\label{eta_real}
    \abs{p_L-p_L^P}\leq\eta,
\end{equation}
since $p_L$ and $p_L^P$ are the probabilities that the two output states assign to the same measurement outcome, and the trace distance bounds this difference for any outcome. This also applies to trap circuits. Denoting by $p^P(\mathcal{T}_i)$ the rate at which an idealized, perfectly twirled trap flips its deterministic output, and $p(\mathcal{T}_i)$ the actual rate that can be measured in an experiment, we have
\begin{equation}\label{eta_traps}
    \abs{p(\mathcal{T}_i)-p^P(\mathcal{T}_i)}\leq\eta.
\end{equation}
The same $\eta$ applies in both equations, since target and trap circuits share the same set of circuit locations except at the resource-preparation slots, which is governed by A2.

Using Eqs. \eqref{eta_real} and \eqref{eta_traps}, we find the upper bound to the actual error rate $p_L$ arising from a stochastic Pauli channel plus a residual coherent contribution:
\begin{align}\label{theorem_generalization}
    p_L&\leq p_L^P+\eta\leq\dfrac{1}{2}\sum_{i=1}^{3K} p^P(\mathcal{T}_i)+\eta \notag \\
    &\leq\dfrac{1}{2}\sum_{i=1}^{3K} \left[p(\mathcal{T}_i)+\eta\right]+\eta=\bar p_L+\left(1+\dfrac{3K}{2}\right)\eta.
\end{align}

Equation \eqref{theorem_generalization} generalizes Theorem \ref{main_theorem} to arbitrary noise channels. The bound survives the failure of A1, acquiring an additive correction rather than failing, and Theorem \ref{main_theorem} is recovered for $\eta=0$. The quantity being bounded, $p_L$, is now the logical error rate of the target circuit under the noise the device actually exhibits, without the assumption that it is exactly Pauli. Its right-hand side retains the same content as before, since $\bar p_L$ is assembled from the trap failure rates measured on the device, and no additional executions are required. What the correction demands is an independent characterization of $\eta$, that is, quantifying how much the actual error channel differs from its ideal Pauli-twirled version. Certification therefore remains possible in the presence of residual coherent errors, at the price of complementing the trap data with a characterization of the residual. Moreover, since $\eta$ appears multiplied by the logical footprint, the characterization step must establish $K\eta$ to be small compared with the error tolerance being certified, so that the upper bound remains informative.

\section{Correction-bitstring statistics and the proof of Lemma 2}\label{lemma2}

In this appendix, we prove that, under assumptions A1--A4 and the trap construction of
Sec.~\ref{traps_subsection}, the target execution and every trap execution suffer
fault configurations, correction bitstrings, and decoder responses with
identical statistics at every location except the resource-preparation slots.
The proof certifies the three required elements: the fault statistics, the
correction bitstrings, and the decoder responses.

\emph{Fault statistics.}
Once the correction bitstring has been fixed, every location other than the resource-state
preparations carries the same physical instruction in the target and in every
trap: gate layers, idle periods, and measurement locations are identical by
construction of the trap circuits, which share the target's circuit
location set. The data-block preparations differ only in which stabilizer
basis state is prepared, and by A3 the preparation noise is independent of
this choice. Under A1, the noise channel attached to any location is fixed
independently of the state passing through it, a consequence of linearity,
since a location-attached channel acts identically on any input. The fault probabilities therefore agree location by location, and, since by A1 the fault distribution is determined by the sequence of executions, the distributions over full fault configurations coincide.

\emph{correction bitstrings.}
In a trap execution, the correction bitstring $\mathbf{c}$ is sampled
uniformly (or equivalently, derandomized over the $2^{k_T}$ correction bitstrings. Each trap is therefore sampled uniformly in its $2^{k_T}$
deterministic variants. Thus, it must be shown that
the target follows the same law. In the target, the recorded outcome can be written as
\begin{equation}\label{pattern_uniformity}
    a_j = \beta_j \oplus e_j,
\end{equation}
where $\beta_j$ is the intrinsic teleportation randomness and
$e_j \in \{0,1\}$ is the flip induced by the accumulated faults and the
decoder's response up to the measurement of ancilla $j$. The bit $\beta_j$
is uniform even in the absence of noise, since just before its measurement, the
injection ancilla is, at the logical level, maximally mixed in the measured
basis, so both outcomes occur with probability $1/2$. Moreover, $\beta_j$
originates from the quantum measurement itself and is independent of the
fault configuration and of all other teleportation bits. The XOR of a
uniform bit with any independent bit is uniform and independent of it, so
each $a_j$ is uniform and independent of the faults, and the full correction
pattern $\mathbf{a} = (a_1, \dots, a_{k_T})$ is uniform on
$\{0,1\}^{k_T}$ and independent of $F$, exactly the law of the sampled $\mathbf{c}$. Consequently, the correction slots carry the same instruction ($S$ or the identity) with the same statistics in both executions. If
corrections are tracked in software rather than applied physically, every
correction slot is idle in all executions and this step is trivial.

\emph{Decoder responses.}
The syndromes produced by a Pauli fault configuration depend only on the
faults and on the correction bitstring, never on the logical state being acted
on, and by A4 a single fixed decoder processes them in all executions. Given
the agreement of fault statistics and correction bitstrings established above,
the decoder responses therefore agree in distribution as well. Finally,
since target and trap shots are interleaved in random order, both sample the
same time-averaged noise, so all of the above refers to a single, common
noise realization.

Target and trap executions thus differ statistically only at the
resource-preparation slots, where the target injects magic states and the
traps stabilizer states. By A2, the injected magic states are no noisier
than the substituted stabilizer states.
\qed

\section{Calibration of the certification test}\label{test_calibration}

As detailed in Sec.~\ref{cert_protococol}, the certification protocol requires balancing the security bound $\delta$ and the yield bound $\gamma$ based on the operational target ratio $r$. Crucially, the determination of the optimal allowed failures $w$ is a scale-invariant problem, as it depends solely on the predefined statistical targets $\delta$ and $\gamma$ and the operational ratio $r$, remaining completely independent of the predefined logical error tolerance $p_{L}^*$.

Algorithm \ref{w_optimization} determines $w$ by evaluating the expected number of errors, $\lambda^*$ and $\lambda^t$, using the inverse $\chi^2$-CDF. These correspond, respectively, to the expected failure counts at the security threshold, Eq.~\eqref{security}, and at the operational target, Eq.~\eqref{yield}. Because the ratio $\lambda^t/\lambda^*$ is smaller than 1 for $w=0$, and monotonically increases with $w$, the procedure simply searches for the minimum integer $w$ that satisfies the operational constraint $\lambda^t/\lambda^* \ge r$. Only at the final step is $p_{L}^*$ introduced, in order to compute the total number of certification shots $M_C$ from Eq.~\eqref{qstar}.

\begin{algorithm}
\caption{Determine the optimal tolerance and number of certification shots}
\label{w_optimization}
\begin{algorithmic}[1]
\Require Global tolerance $p_{L}^*$, logical footprint $K$, significance level $\delta$, false-negative tolerance $\gamma$, operational ratio $r$
\Ensure Optimal tolerance $w$ and required certification shots $M_C$
\State $w \gets 0$
\While{true}
    \State $\lambda^* \gets \dfrac{1}{2}
        F^{-1}_{\chi^2}\!\left(1-\delta,\;2(w+1)\right)$
    \State $\lambda^t \gets \dfrac{1}{2}
        F^{-1}_{\chi^2}\!\left(\gamma,\;2(w+1)\right)$

    \If{$\dfrac{\lambda^t}{\lambda^*}
        \ge r$}
        \State \textbf{break}
    \EndIf
    \State $w \gets w + 1$
\EndWhile
\State $f^* \gets \dfrac{2p_{L}^*}{3K}$
\State $M_C \gets \left\lceil\dfrac{\lambda^*}{f^*}\right\rceil$
\State \Return $(w, M_C)$
\end{algorithmic}
\end{algorithm}

\section{Asymptotic scaling of the certification shot overhead}\label{shot_scaling}

For the protocol described in Sec.~\ref{cert_protococol}, the total shot overhead is governed by the parameter $\lambda^*$, which represents the expected number of aggregated failures at which the acceptance probability hits the significance level $\delta$. It is mathematically defined by the relation $F_P(w;\lambda^*)=\delta$, i.e. by the saturation of the security condition of Eq.~\eqref{security}.

To determine the asymptotic scaling behavior of the shot overhead, we can make use of the relationship between the Poisson and $\chi^2$-CDFs \cite{Johnson_Kotz_Kemp_2005}:
\begin{equation}\label{poisson-chi}
    F_P(w;\lambda) = 1 - F_{\chi^2}(2\lambda; 2(w+1)),
\end{equation}
where $F_{\chi^2}(2\lambda; 2(w+1))$ is the $\chi^2$-CDF with $2(w+1)$ degrees of freedom evaluated at $2\lambda$. The expected number of failures at the security threshold can then be obtained by evaluating the inverse $\chi^2$-CDF:
\begin{equation}
    \lambda^*=\dfrac{1}{2}F^{-1}_{\chi^2}(1-\delta;2(w+1)).
\end{equation}

To study the asymptotic scaling behavior for a large number of allowed failures ($w \gg 1$), we apply the normal approximation of the $\chi^2$-distribution for large degrees of freedom ($2(w+1) \gg 1$). The inverse CDF expands as $F_{\chi^2}^{-1}(1-p; 2(w+1)) \approx 2(w+1) + Z_{1-p}\sqrt{4(w+1)}$, where $Z_{1-p}$ is the $(1-p)$-quantile of the standard normal distribution \cite{Johnson_Kotz_Balakrishnan_1994}. Thus, $\lambda^*$ expands as:
\begin{equation}\label{lambda_expansion}
    \lambda^* \approx (w + 1) + Z_{1-\delta}\sqrt{w + 1}.
\end{equation}
This reveals that the parameter driving the total shot count scales roughly linearly with the allowed tolerance $w$.

However, we do not take $w$ as a free parameter, but it is dictated by the requirement to simultaneously satisfy the security bound $\delta$ and the yield bound $\gamma$ at the operational ratio $r=\lambda^t/\lambda^*<1$, as shown in Appendix \ref{test_calibration}. Applying the normal approximation to both parameters yields:
\begin{equation}
    r\approx\frac{w - Z_{1-\gamma}\sqrt{w}}{w + Z_{1-\delta}\sqrt{w}}.
\end{equation}
From this equation we can isolate how the number of allowed failures scales with the operational ratio as $w\sim \mathcal{O}\left(1/(1-r)^2\right)$. The number of allowed failures quantifies the resolution penalty of the protocol. As the target hardware error rate approaches the tolerance limit $(r\to1)$, the required $w$ and consequently the total number of shots diverges towards infinity. Substituting the scaling of $w$ back in $M_C=\lceil\lambda^*/f^*\rceil$, and using the scaling of normal quantiles $Z_{1-p}\sim \mathcal{O}(\sqrt{-\ln p})$ \cite{Johnson_Kotz_Balakrishnan_1994}, we obtain the ultimate resource scaling of the certification protocol:
\begin{equation}
    M_{C}\sim O\left(\dfrac{K}{p_L^*}\dfrac{(\sqrt{-\ln\delta}+\sqrt{-\ln\gamma})^2}{(1-r)^2}\right).
\end{equation}

\bibliography{ref}

\end{document}